\documentclass[11pt]{article}

\usepackage{amsmath}
\usepackage{amsfonts}
\usepackage{amsthm}
\usepackage{amssymb}

\usepackage{pgfplots}
\pgfplotsset{compat=1.13}

\usepackage{enumerate}
\usepackage{hyperref}
\usepackage{fullpage}

\newtheorem{theorem}{Theorem}

\newtheorem{corollary}[theorem]{Corollary}
\newtheorem{definition}[theorem]{Definition}
\newtheorem{lemma}[theorem]{Lemma}

\newtheorem{problem}[theorem]{Problem}

\newcommand{\E}{\mathop{\mathbb{E}}}

\newcommand{\ApxCount}{\mathsf{ApxCount}_{N,w}}
\newcommand{\AndApxCount}{\mathsf{AND}_2 \circ \mathsf{ApxCount}_{N,w}}

\usepackage[
    backend=bibtex,
    style=numeric-comp,
    natbib=true,
    sortlocale=en_US,
    url=false, 
    doi=true,
    eprint=true,
    isbn=false
]{biblatex}

\begin{document}

\title{$\mathsf{QMA}$ Lower Bounds for Approximate Counting}
\author{William Kretschmer\thanks{University of Texas at Austin. \ Email:
kretsch@cs.utexas.edu. \ Supported by a Simons Investigator award.}}
\date{}
\maketitle

\begin{abstract}
We prove a query complexity lower bound for $\mathsf{QMA}$ protocols that solve approximate counting: estimating the size of a set given a membership oracle. This gives rise to an oracle $A$ such that $\mathsf{SBP}^A \not\subset \mathsf{QMA}^A$, resolving an open problem of Aaronson \cite{1808.02420}. Our proof uses the polynomial method to derive a lower bound for the $\mathsf{SBQP}$ query complexity of the $\mathsf{AND}$ of two approximate counting instances. We use Laurent polynomials as a tool in our proof, showing that the ``Laurent polynomial method'' can be useful even for problems involving ordinary polynomials.
\end{abstract}

\section{Introduction}
Among counting complexity classes, the complexity class $\mathsf{SBP}$ captures approximate counting: estimating a $\mathsf{\#P}$ function within a constant multiplicative factor. Despite having a definition in terms of counting complexity, $\mathsf{SBP}$ is known to lie between two interactive proof classes. In particular, Bohler et. al. \cite{bohler}, who defined $\mathsf{SBP}$, showed that $\mathsf{MA} \subseteq \mathsf{SBP} \subseteq \mathsf{AM}$. Thus, under plausible derandomization assumptions \cite{amnp}, one would have $\mathsf{NP} = \mathsf{MA} = \mathsf{SBP} = \mathsf{AM}$.


In this work, we study the relation between $\mathsf{SBP}$ and $\mathsf{QMA}$. The containment $\mathsf{SBP} \subseteq \mathsf{QMA}$ would follow trivially if $\mathsf{SBP}$ collapses to $\mathsf{MA}$, but it is unclear whether quantum Merlin makes proving this containment any easier. In the relativized world, Aaronson \cite{1808.02420} recently asked whether there might exist an oracle $A$ relative to which $\mathsf{SBP}^A \not\subset \mathsf{QMA}^A$. He noted that exhibiting such an oracle is equivalent to ruling out a black box $\mathsf{QMA}$ protocol for approximate counting. We formally define the approximate counting problem as follows:

\begin{problem}
The approximate counting problem $\ApxCount$ is: given a membership oracle for a set $A \subseteq [N] = \{1, 2, \ldots, N\}$ promised that either $|A| \le w$ (``no'' instance) or $|A| \ge 2w$ (``yes'' instance), determine which of these is the case.
\end{problem}

More generally, one can consider the problem of distinguishing $|A| \le w$ or $|A| \ge (1 + \epsilon)w$ where $\epsilon$ is an arbitrary constant, or may even depend on $N$ and $w$. However, we restrict our attention to fixed $\epsilon$ because $\mathsf{SBP}$ precisely captures approximate counting in the case where $A$ is the set of accepting paths of a nondeterministic polynomial-time Turing machine (and so $|A|$ is a $\mathsf{\#P}$ function), $w$ is an $\mathsf{FP}$ function, and $\epsilon = 1$. Thus, an $\mathsf{SBP}$-$\mathsf{QMA}$ oracle separation would follow if for some function $w(N)$, any $\mathsf{QMA}$ protocol for $\ApxCount$ requires either a $(\log N)^{\omega(1)}$-size witness, or else $(\log N)^{\omega(1)}$ queries.

To prove such a lower bound, we study the query complexity of $\ApxCount$ in the context of the complexity class $\mathsf{SBQP}$, a quantum analogue of $\mathsf{SBP}$ first defined by Kuperberg \cite{kuperberg}. $\mathsf{SBQP}$ is in some sense the smallest ``natural'' complexity class that contains both $\mathsf{SBP}$ and $\mathsf{QMA}$. Indeed, just as $\mathsf{MA}^A \subseteq \mathsf{SBP}^A$ for any oracle $A$, so is $\mathsf{QMA}^A \subseteq \mathsf{SBQP}^A$ for any oracle $A$.

Note that one cannot hope to prove a nontrivial $\mathsf{SBQP}$ query complexity lower bound for approximate counting, as the $\mathsf{SBQP}$ query complexity of $\ApxCount$ is $O(1)$ for any $N$ and $w$. Instead, we use the observation that $\mathsf{SBP}$ is not obviously closed under intersection\footnote{There even exists an oracle relative to which $\mathsf{SBP}$ is \textit{not} closed under intersection \cite{goos}, and $\mathsf{SBP}$'s closure or non-closure under intersection in the unrelativized world remains an open problem.}. In this light, we consider the analogous intersection problem $\AndApxCount$ wherein we are given a pair of sets $A_0, A_1 \subseteq [N]$ and asked to determine whether both sets have size at least $2w$, or whether one of the sets has size at most $w$\footnote{As a technicality, we typically assume that both sets satisfy the $\ApxCount$ promise, though strictly speaking only the smaller set needs to satisfy the promise on a ``no'' instance of $\AndApxCount$.}. Because $\mathsf{QMA}$ is closed under intersection, a $\mathsf{QMA}$ protocol for $\ApxCount$ that receives a witness of size $(\log N)^{O(1)}$ and makes $(\log N)^{O(1)}$ queries implies (via in-place amplification) the existance of an $\mathsf{SBQP}$ algorithm for $\AndApxCount$ that makes $(\log N)^{O(1)}$ queries.

Our main result is that no such $\mathsf{SBQP}$ algorithm exists. Specifically, we show that any $\mathsf{SBQP}$ algorithm for $\AndApxCount$ requires $\Omega\left(\min\left\{\sqrt{w},\sqrt{N/w}\right\}\right)$ queries. We also modify this argument to show that $\Omega(w)$ queries are necessary when $N = 2^{\Omega(w)}$. This in turn shows that any $\mathsf{QMA}$ protocol for $\ApxCount$ that receives a witness of size $m$ and makes $T$ queries must satisfy $m \cdot T = \Omega\left(\min\left\{\sqrt{w},\sqrt{N/w}\right\}\right)$ (or $m \cdot T = \Omega(w)$ when $N = 2^{\Omega(w)}$). Our proof uses the celebrated polynomial method of Beals et. al. \cite{beals}: for an algorithm that makes $T$ queries, we construct a bivariate polynomial $p(x, y)$ of degree at most $2T$ that equals the probability that the algorithm accepts on a random $\AndApxCount$ instance where $|A_0|$ and $|A_1|$ are of fixed size. We then show that if the algorithm is an $\mathsf{SBQP}$ algorithm that correctly solves $\AndApxCount$, then any such polynomial must have large degree.

In our view, the proof of this degree lower bound (Theorem \ref{thm:L}) is of independent mathematical interest. At a high level, from this polynomial $p(x, y)$, we take a parametric curve through the $xy$ plane to construct a univariate Laurent polynomial $q(t)$ of the same degree\footnote{A Laurent polynomial $q(t)$ can contain both positive and negative integer powers of $t$. Formally, we can write $q(t) = q_+(t) + q_-(1/t)$ where $q_+$ and $q_-$ are ordinary polynomials. We follow the convention that the degree of a Laurent polynomial $q(t)$ is $\deg q = \max \{\deg q_+, \deg q_-\}$. This is for consistency with the definition of degree for multivariate polynomials, as in the polynomial $q_+(x) + q_-(y)$ (i.e. viewing $x=t$ and $y=1/t$ as separate indeterminates).}. Crucially, we leverage the symmetries of the problem to view this Laurent polynomial as an ordinary univariate polynomial of the same degree. Finally, we appeal to classical results in approximation theory to argue that this univariate polynomial must have large degree. We find this application of Laurent polynomials surprising, particularly because the recent result of Aaronson on the $\mathsf{BQP}$ query complexity of approximate counting in the QSamples+queries model \textit{also} used Laurent polynomials, albeit for an entirely different reason \citep{1808.02420}. For Aaronson's result, Laurent polynomials are fundamentally necessary just to describe the acceptance probability of the algorithm, while in our case ordinary polynomials suffice. This suggests that the ``Laurent polynomial method'' may prove to be useful even for problems involving ordinary polynomials.

\section{Preliminaries}
Though $\mathsf{SBP}$ and $\mathsf{SBQP}$ can be defined in terms of counting complexity functions (as above), for our purposes it is easier to work with the following equivalent definitions (see B\"ohler et. al. \cite{bohler}):

\begin{definition}\label{def:sbp_sbqp}
The complexity class $\mathsf{SBP}$ consists of the languages $L$ for which there exists a probabilistic polynomial time algorithm $M$ and a polynomial $\sigma$ with the following properties:
\begin{enumerate}
\item If $x \in L$, then $\Pr\left[M(x) \text{ accepts}\right] \ge 2^{-\sigma(|x|)}$.
\item If $x \not\in L$, then $\Pr\left[M(x) \text{ accepts}\right] \le 2^{-\sigma(|x|)}/2$.
\end{enumerate}
The complexity class $\mathsf{SBQP}$ is defined analogously, wherein the classical algorithm is replaced with a quantum algorithm.
\end{definition}
A classical (respectively, quantum) algorithm that satisfies the above promise for a particular language will be referred to as an $\mathsf{SBP}$ (respectively, $\mathsf{SBQP}$) algorithm throughout this paper. Using this definition, a tight query complexity relation between $\mathsf{QMA}$ protocols and $\mathsf{SBQP}$ algorithms follows from the procedure of Marriott and Watrous \cite{marriott}, which shows that one can exponentially improve the soundness and completeness errors of a $\mathsf{QMA}$ protocol without increasing the witness size (see Aaronson \cite{aaronson_szk} for a proof of the following lemma):

\begin{lemma}[Guessing lemma]\label{lem:guessing}
Suppose $V^A$ is a $\mathsf{QMA}$ verifier for some problem and that $V^A$ makes $T$ queries to an oracle $A$, receives an $m$-qubit witness, and has soundness and completeness errors $1/3$. Then there is an $\mathsf{SBQP}$ algorithm $Q^A$ for the same problem that receives no witness and makes $O(m \cdot T)$ queries.
\end{lemma}

Because we study oracle intersection problems, it is often convenient to think of an algorithm as having access to \textit{two} oracles, wherein the first bit in the oracle register selects the choice of oracle. As a consequence, we need a slight generalization of a now well-established fact in quantum complexity: that the acceptance probability of a quantum algorithm with an oracle can be expressed as a polynomial in the bits of the oracle string.

\begin{lemma}[Symmetrization with two oracles]\label{lem:partial_sym}
Suppose $Q^{A_0,A_1}$ is a quantum algorithm that makes $T$ queries to a pair of oracles $A_0, A_1 \subseteq [N]$. Then there exists a bivariate real polynomial $p(x, y)$ of degree at most $2T$ such that:
$$p(x, y) = \E_{\substack{|A_0|=x,\\|A_1|=y}}\left[\Pr[ Q^{A_0,A_1} \text{ accepts}] \right]$$
for all $x, y \in [N]$.
\end{lemma}
\begin{proof}
We can equivalently view the oracles as strings in $\{0,1\}^N$ such that the algorithm makes queries to a single oracle $A = A_0|A_1$ which is the concatenation of the two oracles. Then, Lemma 4.2 of Beals et. al. \cite{beals} tells us that there exists a real polynomial $r(A)$ of degree at most $2T$ such that $r(A) = r(A_0, A_1) = \Pr[ Q^{A_0,A_1} \text{ accepts}]$ for any $A \in \{0,1\}^{2N}$ that is a string of $\{0,1\}$ variables. We then apply the symmetrization lemma of Minsky and Papert \cite{perceptrons} to symmetrize $r$, first with respect to $A_0$, then with respect to $A_1$:
$$p_0(x, A_1) = \E_{|A_0|=x} r(A_0,A_1) =  \E_{|A_0|=x}\left[\Pr[ Q^{A_0,A_1} \text{ accepts}] \right]$$
$$p(x, y) = \E_{|A_1|=y} p_0(x,A_1) = \E_{\substack{|A_0|=x,\\|A_1|=y}}\left[\Pr[ Q^{A_0,A_1} \text{ accepts}] \right]$$
\end{proof}

We now state some useful facts from approximation theory that will be useful in our proofs. We start with the Markov brothers' inequality:
\begin{lemma}[Markov]\label{lem:markov}
Let $p$ be a real polynomial of degree $d$, and suppose that:
$$\max_{x,y\in[a,b]}|p(x)-p(y)| \le H.$$
Then for all $x \in [a, b]$, the derivative $p'$ satisfies:
$$|p'(x)| \le \frac{H}{b-a}d^2.$$
\end{lemma}

This lemma has a useful consequence:
\begin{corollary}\label{cor:markov}
Let $p$ be a real polynomial of degree $d$, and suppose that $|p(x)| \le 1$ for all integers $x \in \{0,1,\ldots,k\}$. If $\max_{x \in [0,k]}|p(x)| \ge 1.001$, then $d =  \Omega(\sqrt{k})$.
\end{corollary}
\begin{proof}
Without loss of generality, we may scale $p$ by some constant and choose $x$ so that $|p(x)| = 1.001$ is the maximum absolute value of $p(x)$ on $[0, k]$. By the mean value theorem, there exists some $x^* \in [\lfloor x \rfloor, \lceil x \rceil]$ such that $|p'(x^*)| \ge 0.001$. Applying the previous lemma, we find that:
\begin{align*}
0.001 &\le \frac{2 \cdot 1.001}{k}d^2\\
\sqrt{\frac{0.001}{2.002}k} &\le d.
\end{align*}
\end{proof}

Put another way, if a polynomial is bounded at all integers $\{0,1,\ldots,k\}$ and has degree $o(\sqrt{k})$, then the polynomial satisfies a marginally weaker bound on all of $[0,k]$. We might wonder whether we can still assume some nontrivial bound when $d$ is not so much smaller than $k$. Indeed we can:

\begin{lemma}[Coppersmith and Rivlin \cite{coppersmith-rivlin}]\label{lem:weak_markov}
Let $p$ be a real polynomial of degree $d \le k$, and suppose that $|p(x)| \le 1$ for all integers $x \in \{0,1,\ldots,k\}$. Then there exist constants $a, b$ that do not depend on $d$ or $k$ such that for all $x \in [0, k]$, we have:
$$|p(x)| \le a \cdot \exp\left(bd^2/k\right).$$
\end{lemma}

We will also use a bound as stated by Paturi \cite{paturi} that bounds a polynomial in terms of its degree and a bound on a nearby interval:
\begin{lemma}\label{lem:chebyshev}
Let $p$ be a real polynomial of degree $d$, and suppose that $|p(x)| \le 1$ for all $|x| \le 1$. Then for all $x$ with $|x| \le 1 + \mu$, we have:
$$|p(x)| \le \exp\left(2d\sqrt{2\mu + \mu^2} \right).$$
\end{lemma}

Setting $\mu = \frac{1}{k}$ and performing some computation gives rise to the following:
\begin{corollary}\label{cor:chebyshev}
Let $p$ be a real polynomial of degree $d$, and suppose that $|p(x)| \le 1$ for all $|x| \le 1$. If $|p(1 + 1/k)| \ge 1.001$, then $d = \Omega(\sqrt{k})$.
\end{corollary}

Finally, we state a useful fact about Laurent polynomials:

\begin{lemma}[Symmetric Laurent polynomials]\label{lem:symmetric}
Let $\ell(x)$ be a real Laurent polynomial of degree $d$ that satisfies $\ell(x) = \ell(1/x)$. Then there exists a real polynomial $q$ of degree $d$ such that $\ell(x) = q(x + 1/x)$.
\end{lemma}

\begin{proof}
$\ell(x) = \ell(1/x)$ implies that the coefficients of the $x^i$ and $x^{-i}$ terms are equal for all $i$, as otherwise $\ell(x) - \ell(1/x)$ would not equal the zero polynomial. Thus, we may write $\ell(x) = \sum_{i=0}^d a_i \cdot (x^i + x^{-i})$ for some coefficients $a_i$. So, it suffices to show that $x^i + x^{-i}$ can be expressed as a polynomial in $x + 1/x$ for all $0 \le i \le d$.

We prove by induction on $i$. The case $i = 0$ corresponds to constant polynomials. For $i > 0$, by the binomial theorem, observe that $(x + 1/x)^i = x^i + x^{-i} + r(x)$ where $r$ is a degree $i - 1$ real Laurent polynomial satisfying $r(x) = r(1/x)$. By the induction assumption, $r$ can be expressed as a polynomial in $x + 1/x$, so we have $x^i + x^{-i} = (x + 1/x)^i - r(x)$ is expressed as a polynomial in $x + 1/x$.
\end{proof}

\section{Main Result}
\subsection{Lower Bound for \texorpdfstring{$\mathsf{SBQP}$}{SBQP}}
Our results hinge on the following theorem, which uses Laurent polynomials to prove a degree lower bound for bivariate polynomials that satisfy a particular set of bounds at points in the plane:

\begin{theorem}\label{thm:L}
Let $w$ and $N$ be integers with $0 < w < 2w \le N$. Let $R_x = [2w,N] \times [0,w]$ and $R_y = [0,w] \times [2w,N]$ be disjoint rectangles in the plane, and let $L = R_x \cup R_y$. Let $p(x, y)$ be a real polynomial of degree $d$ with the following properties:
\begin{enumerate}
\item $p(2w, 2w) \ge 2$.
\item $0 \le p(x, y) \le 1$ for all $(x, y) \in L \cap \mathbb{Z}^2$.
\end{enumerate}
Then $d = \Omega\left(\min\left\{\sqrt{w},\sqrt{N/w}\right\}\right)$.
\end{theorem}
We remark that $L$ gets this name because it looks like the letter ``L'', albeit with the bottom left corner missing (see Figure \ref{fig:L}, shaded regions). The proof idea is as follows. First, we argue that either $d = \Omega(\sqrt{w})$, or else $p$ satisfies a marginally weaker bound on the rectangle $R_x$ by applying the Markov brothers' inequality (via Corollary \ref{cor:markov}) to horizontal and vertical lines through $R_x$. In the latter case, we show that taking an appropriate curve that passes through $R_x$ and the point $(2w, 2w)$ gives rise to a univariate Laurent polynomial $\ell$ of degree $d$. We use Lemma \ref{lem:symmetric} for symmetric Laurent polynomials to reinterpret this as an ordinary polynomial $q$ of degree $d$. We then show that $q$ is bounded on a large interval and grows quickly outside that interval, which implies (by Corollary \ref{cor:chebyshev}) that $q$ has degree $\Omega(\sqrt{N/w})$.

\begin{proof}[Proof of Theorem \ref{thm:L}]
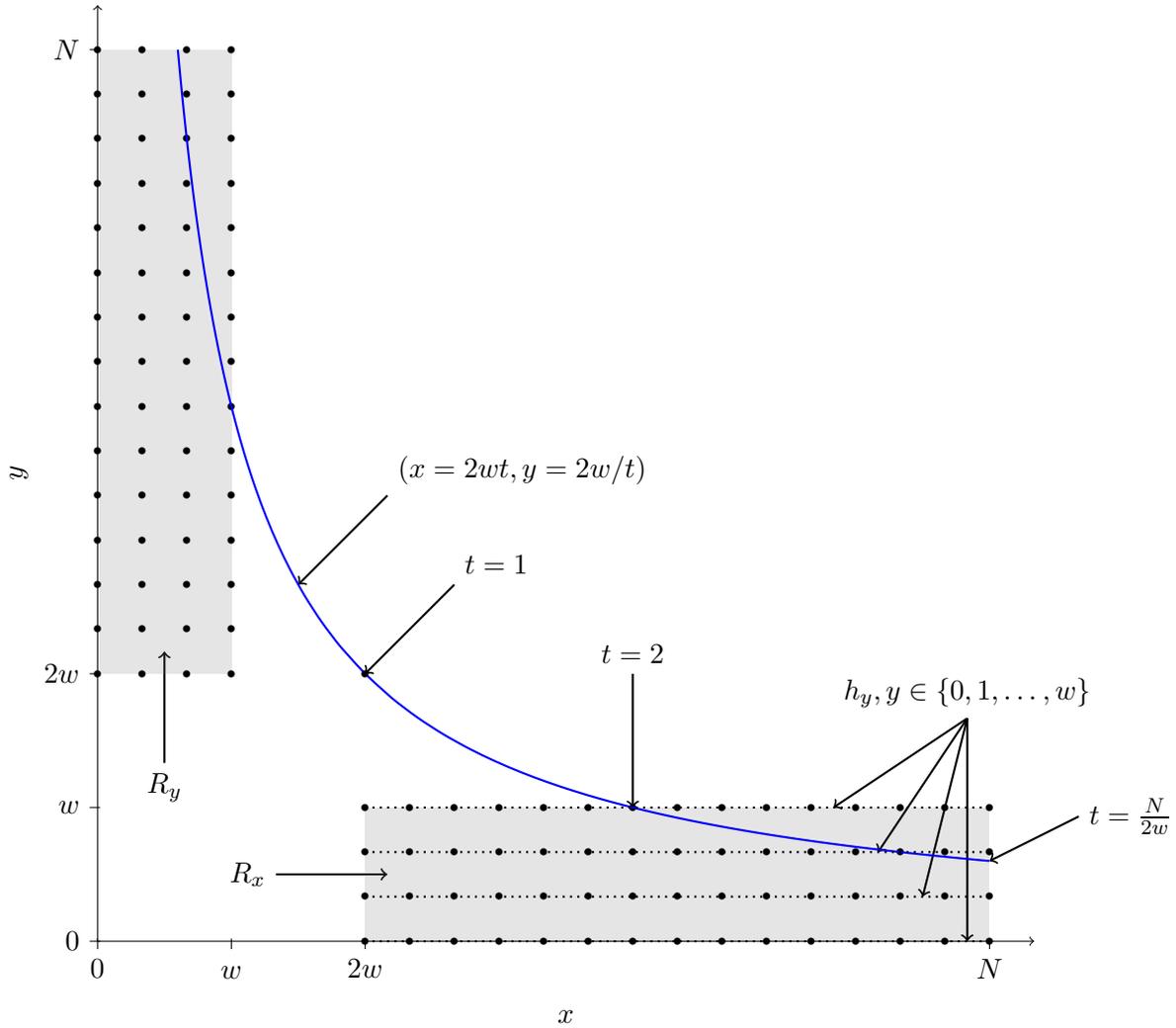
\begin{figure}
\centering
\begin{tikzpicture}[y=0.6cm, x=0.6cm]
	\pgfmathsetmacro{\N}{20};
	\pgfmathsetmacro{\w}{3};
	\pgfmathsetmacro{\tw}{\w+\w};
	\pgfmathsetmacro{\Np}{\N+1};

	\draw[<-,thick] (\tw+0.5, \w/2) -- (\w+1,\w/2) node[anchor=east] {$R_x$};
	\draw[<-,thick] (\w/2, \tw+0.5) -- (\w/2,\w+1) node[anchor=north] {$R_y$};

	\fill[gray,opacity=.2] (\tw,0) -- (\tw,\w) -- (\N,\w) -- (\N,0);
	\fill[gray,opacity=.2] (0,\tw) -- (\w,\tw) -- (\w,\N) -- (0,\N);
	
	\foreach \y in {1,...,\w}
    		\draw[->,thick] (\N-0.5,\w+2) -- (\N-0.5-\y,\y);
    	\draw[<-,thick] (\N-0.5,0) -- (\N-0.5,\w+2) node[anchor=south] {$h_y, y \in \{0,1,\ldots,w\}$};

	\draw[->] (0,0) -- coordinate (x axis mid) (\Np,0);
    	\draw[->] (0,0) -- coordinate (y axis mid) (0,\Np);
    	
	\draw[<-,thick] (\tw,\tw) -- (\tw+2,\tw+2) node[anchor=south west] {$t=1$};   	
	\draw[<-,thick] (\w*4,\w) -- (\w*4,\w*2) node[anchor=south] {$t=2$}; 
	\draw[<-,thick] (\N,4*\w*\w/\N) -- (\N+2,1+4*\w*\w/\N) node[anchor=west] {$t=\frac{N}{2w}$};
	\draw[<-,thick] (1.5*\w,2*\w*4/3) -- (1.5*\w+2,2*\w*4/3+2) node[anchor=south west] {$(x=2wt,y=2w/t)$};
    	
    \draw (0,1pt) -- (0,-3pt) node[anchor=north] {0};
    \draw (\w,1pt) -- (\w,-3pt) node[anchor=north] {$w$\vphantom{l}};
    \draw (\tw,1pt) -- (\tw,-3pt) node[anchor=north] {$2w$};
    \draw (\N,1pt) -- (\N,-3pt) node[anchor=north] {$N$};
    
    \draw (1pt,0) -- (-3pt,0) node[anchor=east] {0};
    \draw (1pt,\w) -- (-3pt,\w) node[anchor=east] {$w$};
    \draw (1pt,\tw) -- (-3pt,\tw) node[anchor=east] {$2w$};
    \draw (1pt,\N) -- (-3pt,\N) node[anchor=east] {$N$};
    
	\node[below=0.8cm] at (x axis mid) {$x$};
	\node[rotate=90, above=0.8cm] at (y axis mid) {$y$};
	
    	\foreach \x in {\tw,...,\N}
    			\foreach \y in {0,...,\w}
     		\node at (\x,\y) {\tiny\textbullet};
    \foreach \x in {0,...,\w}
    			\foreach \y in {\tw,...,\N}
     		\node at (\x,\y) {\tiny\textbullet};
     
     \draw[thick,domain=1:{\N/(2*\w)},smooth,variable=\t,blue] plot ({2*\w*\t},{2*\w/\t});
     \draw[thick,domain=1:{\N/(2*\w)},smooth,variable=\t,blue] plot ({2*\w/\t},{2*\w*\t});
     \node at ({\w+\w},{\w+\w}) {\tiny\textbullet};
     
     \foreach \y in {0,...,\w}
     	\draw[thick,dotted] (\tw,\y) -- (\N,\y);
\end{tikzpicture}
\caption{\label{fig:L}Diagram of Theorem \ref{thm:L}. The lattice points $L \cap \mathbb{Z}^2$ where $0 \le p(x,y) \le 1$ are plotted. Not shown: vertical lines $v_x = x \times [0, w]$ through $R_x$ for each $x \in [2w, N]$ (there are infinitely many such lines).}
\end{figure}

We assume that $3w < N$, as otherwise $\sqrt{N/w} = O(1)$ and the theorem holds trivially.

Let $y \in \{0,1,\ldots,w\}$ be an integer, and consider a horizontal line segment $h_y = [2w, N] \times y$ that passes through $R_x$ (see Figure \ref{fig:L}, dotted horizontal lines). The restriction of $p$ to $h_y$ gives rise to a univariate polynomial $p_y(x)$ of degree $d$. By the assumed bounds on $p(x, y)$ at lattice points in $L$, $|p_y(x)| \le 1$ for all integers $x \in \{2w, 2w + 1,\ldots, N\}$. By the assumption $3w < N$, the interval $[2w, N]$ has length at least $w$. So, we may apply Corollary \ref{cor:markov} to $p_y$ to conclude that either $d = \Omega(\sqrt{w})$, or else $|p(x, y)| < 1.001$ for all $(x, y) \in h_y$.

Now, we use the bounds along the horizontal integer lines through $R_x$ to get bounds along vertical lines. Let $x \in [2w, N]$ (\textit{not necessarily} an integer), and consider a vertical line segment $v_x = x \times [0, w]$ that passes through $R_x$. The restriction of $p$ to $v_x$ gives rise to a univariate polynomial $p_x(y)$ of degree $d$. The intersection of $v_x$ with the $h_y$'s gives a bound $|p_x(y)| < 1.001$ for all integers $y \in \{0,1,\ldots, w\}$. So, we may apply Corollary \ref{cor:markov} to $p_x / 1.001$ to conclude that either $d = \Omega(\sqrt{w})$, or else $|p(x, y)| < 1.001^2$ for all $(x, y) \in v_x$. Because every point $(x, y)$ in the rectangle $R_x$ lies on some $v_x$, we conclude that $|p(x, y)| < 1.001^2$ for all $(x, y) \in R$.

Observe that if $p(x, y)$ satisfies the statement of the theorem, then so does $p(y, x)$. This is because the constraints in the statement of the theorem are symmetric in $x$ and $y$ (in particular, because $R_x$ and $R_y$ are mirror images of one another along the line $x = y$; see Figure \ref{fig:L}). As a result, we may assume without loss of generality that $p$ is symmetric, i.e. $p(x, y) = p(y, x)$. Else, we may replace $p$ by $\frac{p(x, y) + p(y, x)}{2}$ because the set of polynomials that satisfy the inequalities in the statement of the theorem are closed under convex combinations.

Consider the parametric curve $(x = 2wt, y=2w/t)$ as it passes through $R_x$ (see Figure \ref{fig:L}). We can view the restriction of $p(x, y)$ to this curve as a Laurent polynomial $\ell(t) = p(2wt, 2w/t)$ of degree $d$. The bound of $p(x,y)$ on all of $R_x$ implies that $|\ell(t)| < 1.001^2$ when $t \in [2, \frac{N}{2w}]$ and that $\ell(1) \ge 2$ (see Figure \ref{fig:L}). Moreover, the condition that $p(x, y)$ is symmetric implies that $\ell(t) = \ell(1/t)$.

By Lemma \ref{lem:symmetric} for symmetric Laurent polynomials, $\ell(t)$ can be viewed as a degree $d$ polynomial $q(t + 1/t)$. Under the transformation $s = t + 1/t$, $q$ satisfies $|q(s)| < 1.001^2$ for $s \in [2 + 1/2, \frac{N}{2w} + \frac{2w}{N}]$ and $q(2) \ge 2$. Note that the length of the interval $[2 + 1/2, \frac{N}{2w} + \frac{2w}{N}]$ is $\Theta(N/w)$ because $w < N$. By an appropriate affine transformation of $q$, we can conclude from Corollary \ref{cor:chebyshev} with $k = \Theta(N/w)$ that $d = \Omega(\sqrt{N/w})$.
\end{proof}

Theorem \ref{thm:L} implies an $\mathsf{SBQP}$ query complexity lower bound for $\AndApxCount$:

\begin{theorem}\label{thm:sbqp}
Let $Q^{A_0,A_1}$ be an $\mathsf{SBQP}$ algorithm for $\AndApxCount$ that makes $T$ queries to membership oracles $A_0$ and $A_1$. Then $T = \Omega\left(\min\left\{\sqrt{w},\sqrt{N/w}\right\}\right)$.
\end{theorem}
\begin{proof}
Since $Q$ is an $\mathsf{SBQP}$ algorithm, we may suppose that $Q$ accepts with probability at least $2\alpha$ on a ``yes'' instance and with probability at most $\alpha$ on a ``no'' instance. Using Lemma \ref{lem:partial_sym}, take $p(x,y)$ to be the polynomial of degree at most $2T$ that satisfies:
$$p(x, y) = \E_{\substack{|A_0|=x,\\|A_1|=y}}\left[\Pr[ Q^{A_0,A_1} \text{ accepts}] \right].$$
Define $L' = ([0, w] \times [0, w]) \cup ([0, w] \times [2w, N]) \cup ([2w, N] \times [0, w])$. The conditions on the acceptance probability of $Q^A$ for all $A_0, A_1$ that satisfy the $\ApxCount$ promise imply that $p(x, y)$ satisfies these corresponding conditions:
\begin{enumerate}
\item $1 \ge p(x, y) \ge 2\alpha$ for all $(x, y) \in \left([2w, N] \times [2w, N]\right) \cap \mathbb{Z}^2$.
\item $0 \le p(x, y) \le \alpha$ for all $(x, y) \in L' \cap \mathbb{Z}^2$.
\end{enumerate}
In particular, the polynomial $\frac{1}{\alpha} \cdot p(x, y)$ satisfies the (weaker) conditions of Theorem \ref{thm:L}, from which it follows that $T = \Omega\left(\min\left\{\sqrt{w},\sqrt{N/w}\right\}\right)$.
\end{proof}

We remark that even though we could assume $p(x, y) \ge 2\alpha$ over a large region, Theorem \ref{thm:L} only needed $p(x, y) \ge 2\alpha$ at a single point: $(x, y) = (2w, 2w)$. We view this as expressing the intuition that the acceptance probability of an $\mathsf{SBQP}$ algorithm ``should'' be increasing in $|A_0|$ and $|A_1|$.

\subsection{(Non)-Tightness of \texorpdfstring{$\mathsf{SBQP}$}{SBQP} Lower Bound}
In this section, we compare our $\mathsf{SBQP}$ query complexity lower bound for $\AndApxCount$ to known upper bounds. We find a gap between these bounds, particularly when $N$ is much larger than $w$. This motivates an approach to improving our lower bounds for large $N$. In Theorem \ref{thm:L_N_large}, we prove that this approach indeed gives a better lower bound.

The best upper bound we know of for $\mathsf{SBQP}$ query complexity is $O\left(\min\left\{w,\sqrt{N/w}\right\}\right)$, so our bound is at least tight when $N = O(w^2)$. The $O(\sqrt{N/w})$ upper bound follows from the $\mathsf{BQP}$ algorithm of Brassard, H{\o}yer, and Tapp \cite{brassard}. The $O(w)$ upper bound is in fact an $\mathsf{SBP}$ upper bound with the following algorithmic interpretation: first, guess $w + 1$ items randomly from each of $A_0$ and $A_1$. Then, verify using the membership oracle that the first $w + 1$ items all belong to $A_0$ and that the latter $w+1$ items all belong to $A_1$, accepting if and only if this is the case. This accepts with nonzero probability if and only if $|A_0| \ge w + 1$ and $|A_1| \ge w + 1$.

Can the gap between the lower and upper bounds be improved? On the upper bound side, it is tempting to combine Grover search or Brassard-H{\o}yer-Tapp approximate counting with the classical verification to get an $O(\sqrt{w})$ algorithm, but this fails in general because both algorithms always have some nonzero chance of accepting when the number of marked items is nonzero. This suggests that perhaps the lower bound is not tight, at least when $N \gg w$.

Looking for improvements on the lower bound side, careful observation reveals that the main bottleneck in the proof of Theorem \ref{thm:L} is the bound on the growth of polynomials bounded at equally spaced points (Corollary \ref{cor:markov}), which breaks down completely when the polynomial has degree $\omega(\sqrt{w})$. One might observe that we used Corollary \ref{cor:markov} to bound $p(x,y)$ on all of $R_x$, even though we really just need a bound on $p(x,y)$ at the points $(x=2wt, y=2w/t)$.

In fact, this leads to an approach for improving the lower bound, which we now describe. At a high level, we might hope to bound $p(x,y)$ on $(x=2wt, y=2w/t)$ by observing that the curve approaches the line $y=0$ as $t$ grows large. When $N$ is large enough, we \textit{can} still conclude a bound on $p(x, 0)$ for $(x, 0) \in R$ using Corollary \ref{cor:markov}, and intuitively $p(x,y)$ should be close to $p(x,0)$ as $y \to 0$. This intuition indeed works, and allows us to conclude an $\Omega(w)$ lower bound when $N = 2^{\Omega(w)}$. Our strategy for proving this improved lower bound is to show that if $d = o(w)$, then there exists some $\epsilon > 0$ that depends only on $w$ such that $|p(x, y)| < 1.002$ whenever $y \le \epsilon$ and $(x, y) \in R$. Then, the curve $(x = 2wt, y = 2w/t)$ lies in this region whenever $\frac{2w}{\epsilon} \le t \le \frac{N}{2w}$. It follows that the polynomial $q(s)$ as in the proof of Theorem \ref{thm:L} satisfies $|q(s)| < 1.002$ for all $s \in [\frac{2w}{\epsilon} + \frac{\epsilon}{2w}, \frac{N}{2w} + \frac{2w}{N}]$ and $q(2) \ge 2$. So long as $N$ satisfies $\frac{N}{2w} \ge w^2 \cdot \frac{2w}{\epsilon}$, the length of this interval is $\frac{2w}{\epsilon} \cdot \Omega(w^2)$. This gives a contradiction: an appropriate affine transformation of $q$ satisfies the statement of Corollary \ref{cor:chebyshev} with $k = \Omega(w^2)$ but has degree $d = o(w)$. We conclude that $d = \Omega(w)$.

\begin{theorem}\label{thm:L_N_large}
Let $w$, $N$, $L$, $p(x,y)$, and $d$ satisfy the statement of Theorem \ref{thm:L}. If $N = 2^{\Omega(w)}$, then $d = \Omega(w)$.
\end{theorem}

\begin{proof}
Similar to the proof of Theorem \ref{thm:L}, we first bound $p(x, y)$ on horizontal lines through $R_x$, but we can assume a better lower bound because $N$ is now assumed to be large. Exactly as before, we let $y \in \{0,1,\ldots,w\}$ be an integer, and we consider a horizontal line segment $h_y = [2w, N] \times y$ that passes through $R_x$. The restriction of $p$ to $h_y$ gives rise to a univariate polynomial $p_y(x)$ of degree $d$. By the assumed bounds on $p(x, y)$ at lattice points in $L$, $|p_y(x)| \le 1$ for all integers $x \in \{2w, 2w+1,\ldots, N\}$. But now, we can assume $N \gg w^2$, and so Corollary \ref{cor:markov} implies that either $d = \Omega(w)$, or else $|p(x, y)| < 1.001$ for all $(x, y) \in h_y \cap L$.

This time, instead of using Corollary \ref{cor:markov} to bound $p(x, y)$ on vertical lines through $R_x$, we start with the bound of Coppersmith and Rivlin (Lemma \ref{lem:weak_markov}). As before, we let $x \in [2w, N]$ (\textit{not necessarily} an integer), and we consider a vertical line segment $v_x = x \times [0, w]$ that passes through $R_x$. The restriction of $p$ to $v_x$ gives rise to a univariate polynomial $p_x(y)$ of degree $d$. The intersection of $v_x$ with the $h_y$'s gives a bound $|p_x(y)| < 1.001$ for all integers $y \in \{0,1,\ldots, w\}$.

Suppose for a contradiction that $d = o(w)$. Then Lemma \ref{lem:weak_markov} implies that $|p_x(y)| < 2^{o(w)}$ for all $(x, y) \in v_x$. From the Markov brothers' inequality (Lemma \ref{lem:markov}), we can assume that the derivative satisfies $|p'_x(y)| \le \frac{2^{o(w)}}{w} \cdot o(w^2) \le 2^{o(w)}$ for all $y \in [0, w]$. Because $|p_x(0)| < 1.001$, then by basic calculus, there exists $\epsilon = 2^{-o(w)}$ such that $p_x(y) < 1.002$ for all $y \le \epsilon$. In particular, $|p(x, y)| < 1.002$ whenever $y \le \epsilon$ and $(x, y) \in R$.

Recall that it sufficed to show $\frac{N}{2w} \ge w^2 \cdot \frac{2w}{\epsilon}$, or equivalently $N \ge \frac{4w^4}{\epsilon}$ to get a contradiction from the assumption $d = o(w)$. Because $\epsilon = 2^{-o(w)}$, this follows from the assumption that $N = 2^{\Omega(w)}$. We conclude that $d = \Omega(w)$.
\end{proof}

\subsection{Lower Bound for \texorpdfstring{$\mathsf{QMA}$}{QMA}}
We now prove two results about $\mathsf{QMA}$ complexity that follow from the $\mathsf{SBQP}$ lower bound of Theorem \ref{thm:sbqp}:

\begin{corollary}\label{cor:qma_separation}
There exists an oracle $A$ and a pair of languages $L_0, L_1$ such that:
\begin{enumerate}
\item $L_0, L_1 \in \mathsf{SBP}^A$
\item $L_0 \cap L_1 \not\in \mathsf{SBQP}^A$.
\item $\mathsf{SBP}^A \not\subset \mathsf{QMA}^A$.
\end{enumerate} 
\end{corollary}
\begin{proof}
For an arbitrary function $A: \{0,1\}^* \to \{0,1\}$ and $i \in \{0,1\}$, define $A_i^n = \{x \in \{0,1\}^n : A(i, x) = 1\}$. Define the unary language $L^A_i = \{1^n : |A_i^n| \ge 2^{n/2}\}$. Observe that as long as $A$ satisfies the promise $|A_i^n| \ge 2^{n/2}$ or $|A_i^n| \le 2^{n/2-1}$ for all $n \in \mathbb{N}$, then $L^A_i \in \mathsf{SBP}^A$. Intuitively, the oracles $A$ that satisfy this promise encode a pair of $\ApxCount$ instances $|A_0^n|$ and $|A_1^n|$ for every $n \in \mathbb{N}$ where $N = 2^n$ and $w = 2^{n/2 - 1}$. 

Theorem \ref{thm:sbqp} tells us that an $\mathsf{SBQP}$ algorithm $Q$ that makes $o(2^{n/4})$ queries fails to solve $\AndApxCount$ on \textit{some} pair $(A_0, A_1)$ that satisfies the promise. Thus, one can construct an $A$ such that $L_0, L_1 \in \mathsf{SBP}^A$ and $L_0 \cap L_1 \not\in \mathsf{SBQP}^A$, by choosing $(A_0^n, A_1^n)$ so as to diagonalize against all $\mathsf{SBQP}$ algorithms.

Because $\mathsf{QMA}^A$ is closed under intersection for any oracle $A$, and because $\mathsf{QMA}^A \subseteq \mathsf{SBQP}^A$ for any oracle $A$, it must be the case that either $L_0 \not\in \mathsf{QMA}^A$ or $L_1 \not\in \mathsf{QMA}^A$.
\end{proof}

We remark that this gives an alternative construction of an oracle relative to which $\mathsf{SBP}$ is not closed under intersection. To our knowledge, this is the first that uses the polynomial method directly.

Using the guessing lemma (Lemma \ref{lem:guessing}), we can also place an explicit lower bound on the $\mathsf{QMA}$ complexity of $\ApxCount$:

\begin{corollary}\label{cor:qma_tradeoff}
Let $V^A$ be $\mathsf{QMA}$ verifier for the $\ApxCount$ with soundness and completeness errors $1/3$. Suppose $V^A$ receives a witness of length $m$ and makes $T$ queries to a set membership oracle $A$. Then $m \cdot T = \Omega\left(\min\left\{\sqrt{w},\sqrt{N/w}\right\}\right)$.
\end{corollary}
\begin{proof}
Running $V$ a constant number of times with fresh witnesses to reduce the soundness and completeness errors, one obtains a verifier with soundness and completeness errors $1/6$ that receives an $O(m)$-length witness and makes $O(T)$ queries. Repeating twice with two oracles and computing the $\mathsf{AND}$, one obtains a $\mathsf{QMA}$ verifier $V'^{A_0,A_1}$ for $\AndApxCount$ with soundness and completeness errors $1/3$ that receives an $O(m)$-length witness and makes $O(T)$ queries. Applying the guessing lemma (Lemma \ref{lem:guessing}) to $V'$, there exists an $\mathsf{SBQP}$ algorithm $Q^{A_0,A_1}$ for $\AndApxCount$ that makes $O(m \cdot T)$ queries. Theorem \ref{thm:sbqp} tells us that $m \cdot T = \Omega\left(\min\left\{\sqrt{w},\sqrt{N/w}\right\}\right)$.
\end{proof}

Alternatively, one can conclude that $m + T = \Omega\left(\min\left\{w^{1/4},(N/w)^{1/4}\right\}\right)$. Furthermore, when $N = 2^{\Omega(w)}$, one can instead conclude that $m \cdot T = \Omega(w)$ and therefore $m + T = \Omega(\sqrt{w})$ using Theorem \ref{thm:L_N_large} in place of Theorem \ref{thm:L} in the proof of Theorem \ref{thm:sbqp}.

\section{Discussion and Open Problems}
The $\mathsf{QMA}$ lower bound for $\ApxCount$ is not optimal in general: when $w = O(1)$, there is no $\mathsf{QMA}$ protocol for $\ApxCount$ that receives a constant size witness and makes a constant number of queries for large $N$. Fundamentally, this shows that $\mathsf{SBQP}$ lower bounds cannot give optimal $\mathsf{QMA}$ lower bounds. However, our $\mathsf{SBQP}$ bounds themselves are not tight: can one improve the gap between the $\Omega\left(\min\left\{\sqrt{w},\sqrt{N/w}\right\}\right)$ lower bound and $O\left(\min\left\{w,\sqrt{N/w}\right\}\right)$ upper bound for the $\mathsf{SBQP}$ query complexity of $\AndApxCount$? From Theorem \ref{thm:L_N_large}, we know that the complexity must \textit{eventually} reach $\Omega(w)$ (at least when $N$ is exponentially large), so it seems reasonable to conjecture that the $\Omega(\sqrt{w})$ lower bound is not tight even for smaller values of $N$. On the other hand, we have also not made a serious attempt to improve the trivial $O(w)$-query $\mathsf{SBP}$ algorithm using Grover search (or similar techniques). Thus, it appears entirely possible that \textit{neither} bound is tight, perhaps depending on $N$.

At a deeper level, we would like to know if there is any meaningful connection between our use of Laurent polynomials and their use by Aaronson \citep{1808.02420} in studying the QSamples+queries model. One way we might hope to establish such a connection is to extend our proof to an $\mathsf{SBQP}$ lower bound in the QSamples+queries model by proving an analogue of Theorem \ref{thm:L} for Laurent polynomials. We remark that the argument in Theorem \ref{thm:L} that turns a symmetric bivariate polynomial $p(x, y) = p(y, x)$ into a univariate Laurent polynomial $\ell(t) = \ell(1/t)$ works just as well if $p(x, y)$ is a symmetric bivariate Laurent polynomial. Thus, essentially all that is needed is an analogue of Corollary \ref{cor:markov} for Laurent polynomials bounded on a set of equally spaced points\footnote{A priori, it might appear that Theorem \ref{thm:L} also breaks down for large $N$: as $t$ grows, the curve $(x = 2wt, y = 2w/t)$ approaches $y = 0$, where any Laurent polynomial with negative degree terms is clearly unbounded. However, the $\Omega(\sqrt{N/w})$ part of the lower bound really only applies when $N = O(w^2)$, so for the ranges of $N$ that we care about, it is sufficient to look at $t \in [1, 2w]$ over which $y \ge 1$.}.

\section{Acknowledgements}
I would like to thank Thomas Watson for suggesting the intersection approach to proving an $\mathsf{SBP}$-$\mathsf{QMA}$ oracle separation. I am also grateful to Scott Aaronson who supported this research, introduced me to this particular problem, and provided valuable guidance. I would also like to thank Robin Kothari, Justin Thaler, and Patrick Rall for their helpful feedback on this writing.

\printbibliography

\end{document}